\pdfoutput=1
\documentclass{amsproc}
\usepackage{a4wide}
\usepackage{bookmark}

\usepackage{}

\newtheorem{theorem}{Theorem}[section]

\theoremstyle{definition}

\theoremstyle{remark}

\newtheorem{corollary}[theorem]{Corollary}
\numberwithin{equation}{section}

\begin{document}

 \title[Asymptotic expansions for moments of number of comparisons]{Asymptotic expansions for moments of number of comparisons \\  used by the randomized quick sort algorithm}


\author{Sumit Kumar Jha}
\address{Center for Security, Theory, and Algorithmic Research\\ 
International Institute of Information Technology, 
Hyderabad, India}
\curraddr{}
\email{kumarjha.sumit@research.iiit.ac.in}
\thanks{}


\date{}

\begin{abstract}
We calculate asymptotic expansions for the moments of number of comparisons used by the randomized quick sort algorithm using the singularity analysis of certain generating functions.
\end{abstract}

\maketitle

\section{Introduction}
\nocite{prodinger52} 
The basic quick sort algorithm recursively sorts numbers in an array by partitioning the input array into two independent and smaller subarrays, and thereafter sorting these subarrays. The partitioning procedure chooses the first or last element of the array as \emph{pivot} and puts it in its right place so that numbers to the left of it are smaller than it, and those appearing to the right of it are larger than it.\par 
For purposes of this analysis we assume that the input array to the quick sort
algorithm contains \emph{distinct numbers}. Since the sorting algorithm depends only on
the relative order of the numbers in the input array and not their individual values, therefore we can associate to each input array having $n$ elements a permutation of $\{1, 2, \cdots , n\}$.\par 
Let $S_{n}$ denote the set of all permutations of $\{1,\cdots,n\}$. We assume that input arrays of size $n$ are permutations in the set $S_{n}$ with uniform probability distribution (each permutation occurring with probability $1/n!$\, .)\par 
Let $a_{n,k}$ be the number of permutations in $S_{n}$ requiring a total of $k$ comparisons to sort by the quick sort algorithm. Define $s$th factorial moment of the number of comparisons by
\begin{equation}
\label{neq1}
\beta_{s}(n)=\sum_{k\geq 0}(k)_{s}\frac{a_{n,k}}{n!},
\end{equation}
where $(k)_{s}=k(k-1)(k-2)\cdots (k-s+1).$\par 
We wish to calculate the following important result using the singularity analysis of certain generating functions appearing in \cite{flajolet}
\begin{theorem}[\cite{bworld}]
\label{mainthm}
For integers $s\geq 1$ we have
\begin{equation}
\label{neq2}
\beta_{s}(n)=2^{s}n^{s}\log^{s}{n}+2^{s}s(\gamma-2)n^{s}\log^{s-1}{n}+O(n^{s}\cdot \log^{s-2}{n}),
\end{equation}
where $\gamma=0.57721\cdots$ is the Euler's constant.\par 
\end{theorem}
Note that the $s$th moment of number of comparisons admits the same asymptotic expansion as $\beta_{s}(n)$ since it is just a linear combination of $j$th factorial moments for $j\leq s$.
\section{The Calculation}
We start be defining the \emph{probability generating function}:
$$G_{n}(z)=\sum_{k\geq 0}\frac{a_{n,k}\, z^{k}}{n!}.$$
\begin{theorem}
For $n\geq 1$
\begin{equation}
\label{eq1}
G_{n}(z)=\frac{z^{n-1}}{n}\sum_{1\leq j \leq n}G_{n-j}(z)G_{j-1}(z),
\end{equation}
and 
\begin{equation}
\label{eq2}
G_{0}(z)=1.
\end{equation}
\end{theorem}
\begin{proof}
We assume that the first partitioning stage requires $n-1$ comparisons (for some other variants this might be $n+1$ which would result in the same asymptotic expression).\par  
Fix some $k$ such that $1\leq k \leq n$. Let $S_{n}^{k}$ denote the set of all $(n-1)!$ permutations in $S_{n}$ having the pivot element equal to $k$.\par 
Each element $\pi\in S_{n}^{k}$, after the partitioning procedure executed on it, produces a pair $(\sigma_{1}(\pi),\sigma_{2}(\pi))\in S_{k-1}\times S_{n-k}$ where $\sigma_{1}(\pi)\in S_{k-1}$ and $\sigma_{2}(\pi)\in S_{n-k}$ denote the permutations associated with the sub arrays of sizes $k-1$ and $n-k$ obtained after the partitioning procedure is applied on $\pi$, respectively.\par  
We further assume that the partitioning method preserves \emph{equally likely} distribution on \emph{produced permutations}, that is, every pair in $S_{k-1}\times S_{n-k}$ should get produced from partitioning of the same number of permutations in $S^{k}_{n}$, that is, $\binom{n-1}{k-1}$ \cite{hennequin}.
Thus we can write
\begin{equation}
\label{eq3}
a_{n,s}=\sum_{1\leq k \leq n}\binom{n-1}{k-1}\sum_{i+j=s-(n-1)}a_{n-k,i}\,a_{k-1,j}.
\end{equation}
Multiplying equation \eqref{eq3} by $z^{s}$ and dividing by $n!$ we get
\begin{eqnarray*}
\normalfont
\frac{a_{ns}\, z^{s}}{n!}=\sum_{1\leq k \leq n}\frac{z^{s}}{n}\, \sum_{i+j=s-(n-1)}\frac{a_{n-k,i}}{(n-k)!}\cdot \frac{a_{k-1,j}}{(k-1)!}\\
=\sum_{1\leq k \leq n}\frac{z^{s}}{n}\cdot \left\{\text{coefficient of } z^{s-(n-1)} \text{ in } G_{n-k}(z)\cdot G_{k-1}(z)\right\}\\
=\sum_{1\leq k \leq n}\frac{z^{n-1}}{n}\cdot z^{s-(n-1)} \left\{\text{coefficient of } z^{s-(n-1)} \text{ in } G_{n-k}(z)\cdot G_{k-1}(z)\right\}
\end{eqnarray*}
after which summing on $s$ gives us equation \eqref{eq1}.
\end{proof}
Now consider the \emph{double generating function} $H(z,u)$ defined by
\begin{equation}
\label{eq4}
H(z,u)=\sum_{n\geq 0}G_{n}(z) u^{n}.
\end{equation}
\begin{corollary} We have
\begin{equation}
\label{eq5}
\frac{\partial H(z,u)}{\partial u}= H^{2}(z,zu),
\end{equation}
and 
\begin{equation}
\label{eq6}
H(1,u)=(1-u)^{-1}.
\end{equation}
\begin{proof}
From equation \eqref{eq1} we have
\begin{eqnarray*}
\frac{\partial H(z,u)}{\partial u}=\sum_{n\geq 1}(uz)^{n-1}\sum_{1\leq j \leq n}G_{n-j}(z)G_{j-1}(z)\\
=\sum_{n\geq 1}(uz)^{n-1} \cdot \left\{\text{coefficient of } (uz)^{n-1} \text{ in } H(z,uz)\cdot H(z,uz) \right\} \\
=H(z,zu)\cdot H(z,zu).
\end{eqnarray*} 
Equation \eqref{eq6} follows from the fact that $G_{n}(1)=1$.
\end{proof}
\end{corollary}
Now notice that
\begin{equation}
\label{eq7}
\beta_{s}(n)=\left[\frac{d^{s}}{dz^{s}}G_{n}(z)\right]_{z=1}.
\end{equation}
The generating functions $f_{s}(u)$ of $\beta_{s}(n)$ are
\begin{equation}
\label{eq8}
f_{s}(u)=\sum_{n\geq 0}\beta_{s}(n)u^{n}.
\end{equation}
Using Taylor's formula and equation \eqref{eq7} we can write
\begin{equation}
\label{eq9}
H(z,u)=\sum_{s\geq 0}f_{s}(u)\frac{(z-1)^{s}}{s!}.
\end{equation}
\begin{theorem}
For integer $s\geq 0$ we have
\begin{equation}
\label{eq10}
f'_{s}(u)=s!\cdot \sum_{j+k+l+m=s}\frac{f^{(k)}_{j}(u) \cdot f_{l}^{(m)}(u)\cdot u^{k+m} }{j!\cdot k!\cdot  l!\cdot m!},
\end{equation}
\end{theorem}
\begin{proof}
Using Taylor's theorem we can write
$$f_{j}(x)=\sum_{k\geq 0}\frac{f_{j}^{(k)}(u)(x-u)^{k}}{k!}$$
which on substituting $x=uz$ gives
\begin{equation}
\label{eq11}
f_{j}(uz)=\sum_{k\geq 0}\frac{f_{j}^{(k)}(u)(z-1)^{k}u^{k}}{k!}.
\end{equation}
\par 
Now substituting equation \eqref{eq9} in equation \eqref{eq5} gives:
\begin{eqnarray*}
\sum_{s\geq 0}f'_{s}(u)\frac{(z-1)^{s}}{s!}= \sum_{p\geq 0}f_{p}(uz)\frac{(z-1)^{p}}{p!} \cdot \sum_{r\geq 0}f_{r}(uz)\frac{(z-1)^{r}}{r!}\\
= \sum_{p\geq 0}\frac{(z-1)^{p}}{p!}\sum_{l\geq 0}\frac{f_{p}^{(l)}(u)(z-1)^{l}u^{l}}{l!}\cdot  \sum_{r\geq 0}\frac{(z-1)^{r}}{r!}\sum_{m\geq 0}\frac{f_{r}^{(m)}(u)(z-1)^{m}u^{m}}{m!}\\
=\sum_{h\geq 0}(z-1)^{h}\sum_{j+k+l+m=h}\frac{1}{j!}\cdot \frac{f^{(k)}_{j}(u)\, u^{k}}{k!}\cdot \frac{1}{l!}\cdot \frac{f_{l}^{(m)}(u)\, u^{m}}{m!}
\end{eqnarray*}
Now comparing coefficients on both sides of the equation gives
$$f'_{s}(u)=s!\cdot \sum_{j+k+l+m=s}\frac{f^{(k)}_{j}(u) \cdot f_{l}^{(m)}(u)\cdot u^{k+m} }{j!\cdot k!\cdot  l!\cdot m!}.$$
\end{proof}
\noindent
\textbf{Notation:} In the following let $L(u)=\log\left(\frac{1}{1-u}\right).$ Further let $\mathcal{R}_{p,q}(u)$ be an unspecified linear combination of terms of the form $L^{i}(u)(1-u)^{-j}$ where $i,j$ are integers with either $j<q$ and $i$ is arbitrary, or $j=q$ and $i\leq p$.
\begin{corollary}
For integers $s\geq 0$ we have
\begin{equation}
\label{firse5}
f_{s}(u)=\frac{2^{s}\, s!\, L^{s}(u)}{(1-u)^{s+1}}+\frac{s(H_{s}-2)\, 2^{s}\, s!\, L^{s-1}(u)}{(1-u)^{s+1}}+\mathcal{R}_{s-2,s+1}(u),
\end{equation}
where $H_{s}=\sum_{k=1}^{s}\frac{1}{k}$.
\end{corollary}
\begin{proof}
We prove this by induction on $s$. The base case when $s=0$ is trivial since $f_{0}(u)=\sum_{n\geq 0}\beta_{0}(n)u^{n}=(1-u)^{-1}$. Now assume that the assertion is true for all integers $0\leq i\leq s-1$. First observe that the equation \eqref{eq10} is the following linear differential equation:
\begin{equation}
\label{firse1}
f'_{s}(u)-\frac{2}{1-u}f_{s}(u)=p_{s}(u)
\end{equation}
where
\begin{equation}
\label{firse3}
p_{s}(u)=s! \cdot \sum_{\substack{j+k+l+m=s\\ j,l\neq s}}\frac{f^{(k)}_{j}(u) \cdot f_{l}^{(m)}(u)\cdot u^{k+m} }{j!\cdot k!\cdot  l!\cdot m!}.
\end{equation}
Solving the linear differential equation \eqref{firse1} gives us
\begin{equation}
\label{firse2}
f_{s}(u)=\frac{1}{(1-u)^{2}}\int_{0}^{u}p_{s}(t)\cdot (1-t)^{2} dt.
\end{equation}
In the following calculation we quite frequently use the fact that the derivatives $g^{(i)}(u)$ of
$$g(u)=d\cdot q!\cdot L^{p}(u)(1-u)^{-q-1}+\mathcal{R}_{p-1,q+1}(u)\qquad \text{($d$ being a constant)}$$
satisfy
$$g^{(i)}(u)=d\cdot (q+i)!\cdot L^{p}(u)\cdot (1-u)^{-q-i-1}+\mathcal{R}_{p-1,q+i+1}(u).$$ 
Specially for $j\leq s-1$ we have
$$f_{j}^{(i)}(u)=2^{j}\cdot (j+i)!\cdot L^{j}(u)\cdot (1-u)^{-j-i-1}+\mathcal{R}_{j-1,j+i+1}(u).$$
After equation \eqref{firse3} we have
\begin{eqnarray*}
p_{s}(u)=s!\cdot \sum_{\substack{j+l=s\\ j,l\neq s}}\frac{f^{}_{j}(u) \cdot f_{l}^{}(u)}{j!\cdot l!}+2\cdot s!\cdot \sum_{j+l=s-1}\frac{u\cdot f'_{j}(u)\cdot f_{l}(u)}{j!\cdot l!}+\mathcal{R}_{s-2,s+2}(u)\\
=2^{s}\cdot s! \cdot \sum_{\substack{j+l=s\\ j,l\neq s}}\left\{\left(\frac{L^{j}(u)}{(1-u)^{j+1}}+\frac{j(H_{j}-2)L^{j-1}(u)}{(1-u)^{j+1}}\right)\cdot \left(\frac{L^{l}(u)}{(1-u)^{l+1}}+\frac{l(H_{l}-2)L^{l-1}(u)}{(1-u)^{l+1}}\right)\right\}\\
+2^{s}\cdot s! \sum_{j+l=s-1}\left\{ u\cdot (j+1)\frac{L^{j}(u)}{(1-u)^{j+2}}\cdot \frac{L^{l}(u)}{(1-u)^{l+1}}\right\}+\mathcal{R}_{s-2,s+2}(u)\\
=\frac{2^{s}s!}{(1-u)^{s+2}}\left\{(s-1)L^{s}(u)+\left(\frac{s(s+1)}{2}+2\sum_{k=1}^{s-1}k\cdot (H_{k}-2)\right)L^{s-1}(u)\right\}+\mathcal{R}_{s-2,s+2}(u).
\end{eqnarray*}
Thereafter equation \eqref{firse2} gives
\begin{eqnarray*}
f_{s}(u)=\frac{2^{s}\cdot s!}{(1-u)^{2}}\int_{0}^{u}\frac{1}{(1-t)^
{s}}\left\{(s-1)L^{s}(t)+\left(\frac{s(s+1)}{2}+2\sum_{k=1}^{s-1}k\cdot (H_{k}-2)\right)L^{s-1}(t)\right\}dt\\+\mathcal{R}_{s-2,s+1}(u)\\
=\frac{2^{s}s!L^{s}(u)}{(1-u)^{s+1}}+\left(\frac{s}{2}+\frac{2}{s-1}\sum_{k=1}^{s-1}k\cdot (H_{k}-2)\right)\frac{2^{s}s!L^{s-1}(u)}{(1-u)^{s+1}}+\mathcal{R}_{s-2,s+1}(u)\\
=\frac{2^{s}s!L^{s}(u)}{(1-u)^{s+1}}+s\cdot(H_{s}-2)\frac{2^{s}s!L^{s-1}(u)}{(1-u)^{s+1}}+\mathcal{R}_{s-2,s+1}(u)
\end{eqnarray*}
where to deduce the last step we used the identity $\sum_{k=1}^{s-1}k\cdot H_{k}=\frac{(s)(s-1)}{2}\left(H_{s}-\frac{1}{2}\right),$ and to deduce the second last step we used integration by parts.
\end{proof}
\begin{proof}[Proof of Theorem \ref{mainthm}]
We would use the following result to conclude the assertion:
\begin{theorem}[Flajolet and Odlyzko \cite{flajolet}]
Let 
\begin{equation}
f_{\alpha,\beta}(u)\equiv f(u)=\frac{1}{(1-u)^{\alpha}}\left(\log\frac{1}{1-u}\right)^{\beta},
\end{equation}
where $\alpha$ is a positive integer and $\beta$ is a non-negative integer. The coefficient of $u^{n}$ in $f(u)$, denoted $[u^{n}]f(u)$, admits the asymptotic expansion
\begin{equation}
\label{main}
[u^{n}]f(u)\sim \frac{n^{\alpha-1}}{(\alpha-1)!}(\log n)^{\beta}\left[1+\frac{C_{1}}{1!}\,\frac{\beta}{\log n}+\frac{C_{2}}{2!} \, \frac{(\beta)(\beta-1)}{\log^{2}(n)}+\cdots \right],
\end{equation}  
where 
$$C_{k}=(\alpha-1)!\left[\frac{d^{k}}{dx^{k}}\frac{1}{\Gamma(x)}\right]_{x=\alpha}$$
and $\Gamma(\cdot)$ is the gamma function.
\end{theorem}
We can use the above result to conclude that
$$[u^{n}]\left\{2^{s}\cdot s!\cdot \frac{{L}^{s}(u)}{(1-u)^{s+1}}\right\}=2^{s}{n^{s}}(\log^{s}{n}+s\cdot C_{1}\cdot \log^{s-1}{n}+O(\log^{s-2}{n})),$$
$$[u^{n}]\left\{2^{s}\cdot s!\cdot \frac{L^{s-1}}{(1-u)^{s+1}}\right\}=2^{s}n^{s}(\log^{s-1}{n}+O(\log^{s-2}{n})).$$
Here
$$C_{1}=s!\cdot \left[\frac{d}{dx}\frac{1}{\Gamma(x)}\right]_{x=s+1}=-\frac{s!}{\Gamma(s+1)}\cdot \frac{\Gamma'(s+1)}{\Gamma(s+1)}=- (H_{s}-\gamma).$$
Here we used the result $\frac{\Gamma'(s+1)}{\Gamma(s+1)}=H_{s}-\gamma=\psi(s+1)$ where $\psi(\cdot)$ is the digamma function \cite{digamma}. Now we can conclude our result after equation \eqref{firse5}.
\end{proof}

\bibliographystyle{amsplain}
\bibliography{sample.bib}
\end{document}